\documentclass[a4paper,12pt]{article}
\usepackage{geometry}
\usepackage[toc,page]{appendix}
\usepackage[colorlinks,citecolor=blue,linkcolor=blue, urlcolor=blue]{hyperref}
\usepackage{amsmath,amsthm,amssymb}
\usepackage{graphicx}
\usepackage{caption}
\usepackage{subcaption}
\usepackage{cleveref}
\usepackage{sectsty}
\usepackage{mathtools}
\usepackage{color}
\usepackage{dsfont}
\usepackage{comment}
\usepackage[framemethod=tikz]{mdframed}
\usepackage{geometry}
\usepackage{enumitem}
\theoremstyle{definition}
\newtheorem{theorem}{Theorem}[section]

\newtheorem{definition}[theorem]{Definition}

\newtheorem{proposition}[theorem]{Proposition}

\newtheorem{example}[theorem]{Example}

\newtheorem{remark}[theorem]{Remark}

\newcommand*{\argmin}{\operatornamewithlimits{argmin}\limits}
\newcommand{\BE}[1]{\mathbb{E}\left({#1}\right)}
\newcommand{\BP}[1]{\mathbb{P}\left({#1}\right)}

\newcommand{\BR}{\mathbb{R}}

\newcommand{\ve}{\varepsilon}
\newcommand{\norm}[1]{\left\lVert#1\right\rVert}

\sectionfont{\fontsize{12}{15}\selectfont}
\subsectionfont{\fontsize{12}{15}\selectfont}
\numberwithin{equation}{section}
\date{\today}
\makeatletter % so @ will be a letter
\@addtoreset{equation}{section}

\makeatother % so @ will not be a letter
\begin{document}
\title{Risk minimization and portfolio diversification}

\author{Farzad Pourbabaee \footnote{Department of Mathematics and Statistics, McMaster University, 1280 Main Street West, Hamilton, ON, L8S 4K1, Canada, farzad.pourbabaee@math.mcmaster.ca} \and Minsuk Kwak \footnote{Department of Mathematics and Statistics, McMaster University, 1280 Main Street West, Hamilton, ON, L8S 4K1, Canada, minsuk.kwak@gmail.com} \and Traian A.~Pirvu\footnote{Department of Mathematics and Statistics, McMaster University, 1280 Main Street West, Hamilton, ON, L8S 4K1, Canada, tpirvu@math.mcmaster.ca}}

\maketitle

\begin{abstract}
We consider the problem of minimizing capital at risk in the Black-Scholes setting. The portfolio problem is studied  given the possibility that a correlation constraint between the portfolio and a financial index is imposed. The optimal portfolio is obtained in closed form. The effects of the correlation constraint are explored; it turns out that this portfolio constraint leads to a more diversified portfolio.  
\end{abstract}

\section{Introduction}
There is a common sense among financial investors to maximize the portfolio return while satisfying some risk constraint. Mean-variance techniques to address this problem have been developed in \cite{markowitz1991foundations}. The objective in portfolio selection is decreasing the investment downside risk (this risk is quantified through various measures like VaR). These notions in portfolio selection and risk management result in a great deal of literature that has been published yet. For instance, the notion of value at risk, as the $\alpha$ quantile subtracted from the \textit{mean} of portfolio return has been thoroughly investigated in \cite{duffie1997overview, jorion2007value}, which turns out to suffer from being non-coherent risk measure. Capital at risk (CaR) is introduced to obviate this problem, which differs from VaR through a constant (it is VaR adjusted to the riskless return). Some other risk measures such as average value at risk (AVaR) and limited expected loss (LEL) were introduced to address the shortcomings of VaR. Analytical formulas for these types of risk measures, as well as risk constrained portfolio optimization in continuous time framework are provided in \cite{gambrah2014risk}.

Portfolio selection under bounded capital at risk is well explored in \cite{emmer2001optimal}. In a Black-Scholes setting with constant coefficients, they obtained a closed form solution for the optimal portfolio with maximum mean return, which has a bounded CaR. It has been shown in \cite{emmer2001optimal} that the use of merely variance based measures leads to a decreased proportion of risky assets in a portfolio when the planning horizon increases, effect not exhibited when CaR is employed; this argument comes to support the use of CaR in risk management. \cite{emmer2001optimal} has been extended to allow for no-short-selling constraint in \cite{dmitravsinovic2011dynamic}. 
  \cite{dmitravsinovic2006asymptotic} has shown the counterintuitive behaviour in VaR constrained optimized portfolios, due to an increase in investment time horizon T, effect does not appear when CaR is employed.

The literature on portfolio selection subject to a correlation constraint is rather limited. Here we only mention \cite{bernard2014mean}. They studied the mean variance optimal portfolios in the presence of a stochastic benchmark correlated to the market, and discussed how their method could be used to detect fraud in financial reports; for instance under some conditions one could not claim of a positive Sharpe ratio, while having a negative correlation with market index.

Our paper provides in closed form the optimal portfolio for minimizing CaR in Black-Scholes setting, subject to a correlation constraint between an investor's portfolio and an index/target process (modelled as a geometric Brownian motion (GBM)). One possible choice of this index/target is the optimal growth portfolio (see \cite{bernard2014mean}). The addition of the correlation constraint is useful in some situations; for instance, maintaining a negative correlation with the index/target will let us to better control the risky investment during the market crash (that is when the market index is declining heavily); in such situations the negative correlation could rescue our portfolio from falling down. By analyzing the closed form solutions of the constrained and unconstrained portfolio selection problems, we notice that the correlation constraint leads to a more diversified portfolios (we use variance as a measure of diversification).

The rest of the paper is organized as follows: In Section \ref{CaRwithout} we introduce the market model, our specific notion of CaR. Section 3 briefly reiterates the minimization problem of capital at risk. Section \ref{CaRwith} is devoted to the problem of CaR minimization under correlation constraint. A particular market model is considered in Section \ref{CaRc}, which is followed by some numerical examples. The paper concludes in Section \ref{conclusion}. The proofs of our results are delegated to an appendix. 
%%%%%%%%%%%%%%%%%%%%%%%%%%%%%%%%%%%%%%%%%%%%%%%%
\section{Model}
\label{CaRwithout}
Consider a probability space $(\Omega,\{\mathcal{F}_t\}_{0\leq t\leq
T}, {P}),$ which accommodates a standard multidimensional Brownian
motion. Throughout this paper we restrict ourselves to a geometric Brownian motion model for stock prices.
Let us consider a financial market model with the following specifications:
\begin{itemize}[leftmargin=*]
\item Assets are traded continuously over a finite time horizon $[0,T]$ in a frictionless market.

\item There is one riskless asset, denoted by $S_0(t)$, with positive constant interest rate $r$ \begin{equation}\frac{dS_0(t)}{S_0(t)}=r dt.
\end{equation}

\item There are $d$ risky assets (stocks), driven by a $d$-dimensional Brownian motion:
\begin{align}
\label{stocks}
\frac{dS_i(t)}{S_i(t)}=(r+b_i) dt + \sum_{i=1}^{d}\sigma_{ij}dW_j(t),\quad S_i(0)=s_i,\quad i=1,\dots,d.
\end{align}
Here $[\sigma_{ij}]$ is the square volatility matrix, which we assume to be invertible, and $b = (b_1,\ldots,b_d)'$ is the vector of excess return rate of each risky asset, which we assume positive.

\item Let $\pi = (\pi_1,\ldots,\pi_d)' \in \BR^d$ be the portfolio vector of the investor, where $\pi_i$ indicates the fraction of total initial wealth $x$ invested in stock $i$, which is assumed constant. Thus, $\pi_0=1-\mathbf{1}'\pi$ is the fraction of wealth invested in the risk free bond. We draw our attention on constant proportions, which are time invariant, assumption needed for tractability. Constant $\pi$ does not mean there is no trade since one needs to rebalance the portfolio continuously to keep the portion invested in each asset constant over time.

The stochastic dynamic equation of the wealth process is:
\begin{align}
\frac{dX^{\pi}(t)}{X^{\pi}(t)}=(r+b'\pi)dt+\pi'\sigma dW(t),\quad X^{\pi}(0)=x. \label{dX}
\end{align}
Direct computations lead to:
\begin{align*}
&X^{\pi}(T) = x e^{[(r+b'\pi-\norm{\sigma' \pi}^2/2)T+\pi'\sigma W(T)]}\\
&\BE{X^{\pi}(T)} = x e^{(r+b'\pi)T}\\
& \text{Var}(X^{\pi}(T)) = x^2 e^{2(r+b'\pi)T}\left(e^{\norm{\sigma' \pi}^2T}-1\right)\\
& \text{Var}(\log X^{\pi}(T))  = T \pi'\sigma \sigma' \pi.
\end{align*}
\end{itemize}
In order to keep tractability we perform risk measurements for logarithmic returns (rather than arithmetic returns). It is well known that for small
time horizons the two types of returns are close to each other. Let us give a formal definition of CaR.
\begin{definition}[Capital-at-risk]
\label{CaR-def}
Let $z_\alpha$ be the $\alpha$-quantile of the standard gaussian distribution. The CaR of a fixed portfolio vector $\pi$ is defined as the difference of Log riskless return and the $\alpha$-quantile of the Log return over $[0,T]$:
\begin{equation}
\label{CaR}
\begin{aligned}
q(x,\pi,\alpha,T) &= \inf \{ z \in \BR: \BP{\log \left( \frac{X^{\pi}(T)}{X^{\pi}(0)} \right) \leq z} \geq \alpha\}\\
& = (r+b'\pi)T-\frac{1}{2}\norm{\sigma' \pi}^2T+ z_\alpha \norm{\sigma' \pi} \sqrt{T}\\
\text{CaR}(\pi,\alpha,T) &:= rT - q(x,\pi,\alpha,T)\\
&= -b'\pi T+\frac{1}{2}\norm{\sigma' \pi}^2T-z_\alpha \norm{\sigma' \pi} \sqrt{T}.
\end{aligned}
\end{equation}
We assume that $ \alpha < 0.5,$ which means $z_\alpha < 0.$ 
\end{definition}
%%%%%%%%%%%%%%%%%%%%%%%%%%%%%%%%%%%%%%%%%%%%%%%%
\section{CaR Minimization}

CaR minimization is first explored in \cite{emmer2001optimal}, in which a closed form solution is found for the portfolio with maximum expected return under bounded CaR constraint. Next, let us state the proposition of this section which for portfolio arithmetic returns has been established in \cite{emmer2001optimal}. 
\begin{proposition}
\label{CaR-min}
The minimum CaR portfolio, i.e., the solution of
\begin{equation}
\label{riskminu}
\begin{aligned}
\argmin_{\pi\in\BR^d} ~~ &\text{CaR}(\pi,\alpha,T)\\
\end{aligned}
\end{equation}
 is
\begin{align}
\label{opt-pi}
\pi^* &= \left(\frac{z_\alpha}{\sqrt{T}}+\norm{\sigma^{-1}b}\right)^+ \frac{(\sigma \sigma')^{-1}b}{\norm{\sigma^{-1}b}}.
\end{align}
Moreover
\begin{align}
\label{opt-CaR}
\text{CaR}(\pi^*,\alpha,T) &= -\frac{T}{2} \left[\left(\frac{z_\alpha}{\sqrt{T}}+\norm{\sigma^{-1}b}\right)^+\right]^2.
\end{align}
\end{proposition}

\begin{proof}
See the appendix.
\end{proof}

%%%%%%%%%%%%%%%%%%%%%%%%%%%%%%%%%%%%%%%%%%%%%%%%
\section{CaR Minimization under Correlation Constraint}
\label{CaRwith}
In this section we focus on minimizing the CaR subject to a correlation constraint. Namely, we want to find the optimal portfolio which minimizes the CaR, as well as satisfying a correlation constraint with another target/index portfolio. Assume that the target/index portfolio dynamics is given by:
\begin{align}
\frac{dY(t)}{Y(t)}=(r+b'\eta)dt+\eta'\sigma dW(t),\quad Y(0)=y. \label{dY}
\end{align}
Moreover, we assume that the target portfolio has positive excess return over risk-free rate, i.e $b'\eta > 0$, and enforce that the correlation between the Log values of $X(T)$ and $Y(T)$ does not exceed a negative threshold. This condition is expressed as:
\begin{align}
\label{cond}
\text{Corr}(\log X(T),\log Y(T)) \leq -\delta, \quad \text{where} \quad \delta \geq 0. 
\end{align}
The $Y$ process can be any financial index or portfolio process which is driven by the same sources of uncertainty as stocks in the market. 
The correlation between terminal Log values of the processes is found in closed form:
\begin{align}
\text{Corr}(\log X(T),\log Y(T)) = \frac{\pi'\sigma \sigma' \eta}{\norm{\sigma' \pi} \norm{\sigma' \eta}}.
\end{align}
Then the risk minimization problem under correlation constraint would be:
\begin{equation}
\label{riskmin}
\begin{aligned}
\min_{\pi\in\BR^d} ~~ &\text{CaR}(\pi,\alpha,T)\\
\text{subject to} ~~ &\text{Corr}(\log X(T), \log Y(T)) \leq -\delta . \\
\end{aligned}
\end{equation}
Which is equivalent to:
\begin{equation}
\label{optprob}
\begin{aligned}
\min_{\pi\in\BR^d} ~ &\left(-b'\pi T+\frac{1}{2}\norm{\sigma' \pi}^2T-z_\alpha \norm{\sigma' \pi} \sqrt{T}\right)\\
\text{subject to} ~~ & \delta \norm{\sigma' \eta}\norm{\sigma' \pi}+\pi'\sigma \sigma' \eta \leq 0.
\end{aligned}
\end{equation}
The requirement for negative correlation is imposed for tractability reasons (it renders the optimization problem convex). Moreover, it has
a financial interpretation: in times of financial downturn it is desirable to be negatively correlated with the market portfolio.  
Due to the explicit formulation of the problem, the solution can be found analytically with convex optimization methods. The following theorem presents the main result of the paper. 
\begin{theorem}
\label{main}
The optimal portfolio which solves \eqref{optprob} is:
\begin{align}\label{pic}
\pi^*= \frac{\sqrt{1-\delta^2}\left(\frac{z_\alpha \norm{\sigma' \eta}}{\sqrt{T}}+\sqrt{1-\delta^2}\sqrt{\norm{\sigma^{-1}b}^2\norm{\sigma' \eta}^2 -(b' \eta)^2}- \delta b'\eta\right)^+}{T\sqrt{\norm{\sigma^{-1}b}^2\norm{\sigma' \eta}^2 -(b' \eta)^2}}\left((\sigma\sigma')^{-1}bT-\lambda^*\eta\right).
\end{align}
The parameter $\lambda^*$ is given by:
\begin{align}
\lambda^* = \frac{1}{\norm{\sigma' \eta}^2} \left(b' \eta T+\frac{T\delta}{\sqrt{1-\delta^2}}\sqrt{\norm{\sigma^{-1}b}^2\norm{\sigma' \eta}^2 -(b' \eta)^2}\right).
\end{align}
and
\begin{align}
\text{CaR}(\pi^*,\alpha,T) = \frac{-T}{2\norm{\sigma'\eta}^2} \left[\left(\frac{z_\alpha \norm{\sigma' \eta}}{\sqrt{T}}+\sqrt{1-\delta^2}\sqrt{\norm{\sigma^{-1}b}^2\norm{\sigma' \eta}^2 -(b' \eta)^2}- \delta b'\eta\right)^+\right]^2.
\end{align}
\end{theorem}

\begin{proof}
See the appendix.
\end{proof}

\begin{remark}
The optimal portfolio in \eqref{pic} displays a two fund separation structure, first component is similar to the solution in \eqref{opt-pi}, and the second component 
is induced by the correlation constraint. As it can be seen from the closed form expressions, unless the zero portfolio is optimal, the inequality correlation constraint always binds $\lambda^* > 0$, and this  means the optimal portfolio pushes itself all the way to have the correlation with target portfolio equal to $-\delta$.
\end{remark}
%%%%%%%%%%%%%%%%%%%%%%%%%%%%%%%%%%%%%%%%%%%%
\section{Pricing Kernel Inverse as the Benchmark Portfolio}
\label{CaRc}
In this section we study the special case of a target/index portfolio. The assets are regrouped into two parts: the mutual funds accessible to both the investor's portfolio and the target/index portfolio, and some extra risky assets available to the portfolio manager (but not entering into the composition of  target/index portfolio). Let us decompose the Brownian motion vector as $W(t)=(W_1(t)' ~~ W_2(t)')'$, where the first component is  $m$ dimensional. Moreover without any loss of generality we can represent the volatility matrix and its inverse as:
\begin{align}
\label{volmat}
\sigma=\begin{bmatrix}
\sigma_{11} & \mathbf{0}\\
\sigma_{21} & \sigma_{22}
\end{bmatrix} , \quad
\sigma^{-1}=\begin{bmatrix}
\sigma_{11} ^{-1} & \mathbf{0}\\
-\sigma_{22}^{-1}\sigma_{21}\sigma_{11}^{-1} & \sigma_{22}^{-1}
\end{bmatrix}.
\end{align}
Here $\sigma_{11}$ and $\sigma_{22}$ are square $m$ and $d-m$ matrices, respectively, and $\sigma_{21}$ and $\mathbf{0}$ are $(d-m \times m)$ and $(m \times d-m)$ zero matrices, respectively. Consequently, there are two types of stocks: the first type is only driven by $W_1(t)$ and the second type is driven by both components of the Brownian motion. Following the setting from the previous section the manager's portfolio is expressed as:
\begin{align}
\nonumber
\frac{dX(t)}{X(t)}&= (r+b' \pi)dt + \begin{bmatrix} \pi'_1& \pi'_2 \end{bmatrix} \begin{bmatrix}
\sigma_{11} & \mathbf{0}\\
\sigma_{21} & \sigma_{22}
\end{bmatrix} \begin{bmatrix} dW_1(t) \\ dW_2(t)\end{bmatrix}\\
&=(r+b' \pi)dt + (\pi'_1\sigma_{11}+\pi'_2\sigma_{21})dW_1(t)+\pi'_2\sigma_{22} dW_2(t).
\end{align}
Here $ \pi '= \begin{bmatrix} \pi'_1& \pi'_2 \end{bmatrix}$ is the portfolio vector, and $b' = \begin{bmatrix} b_1' & b_2' \end{bmatrix}$ is the excess return vector of the first and second type of stocks respectively.

Following \cite{bernard2014mean} the target/index portfolio is taken to be the optimal growth portfolio, which in our setup is the inverse of pricing kernel. 
Let us re-emphasize that on constructing the market portfolio only the first type of stocks is considered. Hence if $\theta = \sigma_{11}^{-1}b_1$ denotes the market price of risk of first type stocks, the pricing kernel $Z$ is given by
\begin{align}
\frac{dZ(t)}{Z(t)} = -\theta' dW_1(t).
\end{align}
Consequently, the target/index portfolio is $Y(t) = Z(t)^{-1}$, which satisfies the SDE:
\begin{align}
\nonumber
\frac{dY(t)}{Y(t)} &= \norm{\theta}^2 dt+\theta'dW_1(t)\\
&=\norm{\sigma_{11}^{-1}b_1}^2 dt+((\sigma_{11} \sigma_{11}')^{-1}b_1)'\sigma_{11} dW_1(t).
\end{align}
Let us define $\eta' = \begin{bmatrix} ((\sigma_{11} \sigma_{11}')^{-1}b_1)' & \mathbf{0} \end{bmatrix}$ to be the associated portfolio vector for 
the target/index portfolio. Theorem \ref{main} takes the following form in this context. 
\begin{proposition}
\label{pk}
The optimal portfolio which minimizes the capital at risk, and satisfies the correlation constraint is:
\begin{align}
\label{optpipk}
\pi^*= \frac{\sqrt{1-\delta^2}\left(\frac{z_\alpha}{\sqrt{T}}+\sqrt{1-\delta^2}\norm{\sigma_{22}^{-1}b_2-\sigma_{22}^{-1}\sigma_{21}\sigma_{11}^{-1}b_1}- \delta \norm{\sigma_{11}^{-1}b_1}\right)^+}{T\norm{\sigma_{22}^{-1}b_2-\sigma_{22}^{-1}\sigma_{21}\sigma_{11}^{-1}b_1}}\left((\sigma\sigma')^{-1}bT-\lambda^*\eta\right).
\end{align}
Here $\lambda^*$ is given by:
\begin{align}
\label{optlambdapk}
\lambda^*=T\left(1+\frac{\delta}{\norm{\sigma_{11}^{-1}b_1}\sqrt{1-\delta^2}}\norm{\sigma_{22}^{-1}b_2-\sigma_{22}^{-1}\sigma_{21}\sigma_{11}^{-1}b_1}\right).
\end{align}
\end{proposition}

\begin{proof}
See the appendix.
\end{proof}

Closed form expressions of this proposition allow us to explore the impact of correlation constraint on the portfolio diversification. In the following an improvement in portfolio diversification due to correlation constraint is established. To distinguish between the optimal portfolios found in proposition \ref{CaR-min} and theorem \ref{main}, from now on all the optimal variables for the problem \textit{with} correlation constraint are going to be shown by subscript ``$c$'', like $\pi^*_c$. Let us state the result.
\begin{proposition}
\label{var-reduction}
The optimal portfolio of the constrained problem $\pi^*_c$ is more diversified than the optimal portfolio of the unconstrained problem $\pi^*.$
That is to say
\begin{align}
\label{varcompar}
\text{Var}(\log X^{\pi^*}(T)) \geq \text{Var}(\log X^{\pi^*_c}(T)).
\end{align}
\end{proposition}

\begin{proof}
See the appendix.
\end{proof}

%%%%%%%%%%%%%%%%%%%%%%%%%%%%%%%%%%%%%%%%%%%%%%%%%%
\subsection{Diversification and Risk Control over Market Downfalls}
As an application of our portfolio optimization problem with correlation constraint, we would like to explore the diversification during the period of a market collapse. In order to be able to track the downfalls in the market and for the ease of exposition, we assume that the first type of stocks have the uncertainty driven by a single Brownian motion only. Then by letting large enough values of $\sigma_{11}$ (which would happen during a time of market crash), we  consider the asymptotic composition of the optimal portfolios (constrained and unconstrained). Direct computations lead to: 
\begin{align}
\lim_{\sigma_{11} \to \infty} \pi^* &= \left(\frac{z_\alpha}{\norm{\sigma_{22}^{-1}b_2}\sqrt{T}}+1 \right)^+ \begin{bmatrix} 0 \\ (\sigma_{22}\sigma'_{22})^{-1}b_2\end{bmatrix}.\\
\lim_{\sigma_{11} \to \infty} \pi^*_c &=\sqrt{1-\delta^2}\left(\frac{z_\alpha}{\norm{\sigma_{22}^{-1}b_2}\sqrt{T}}+\sqrt{1-\delta^2}\right)^+\begin{bmatrix} 0 \\ (\sigma_{22}\sigma'_{22})^{-1}b_2\end{bmatrix}.
\end{align}
Let us notice the zero investment in the first type of stocks in both constrained and unconstrained optimal portfolios. Correlation constraint on the other hand 
lowers the investment in the second class of stocks. The diversification benefit of the correlation constraint can also be seen from considering the asymptotic variances of $\log X^{\pi^*}(T)$ and $\log X^{\pi^*_c}(T)$, which are computed to be:
\begin{align}
\lim_{\sigma_{11} \to \infty} \text{Var}(\log X^{\pi^*}(T)) &= T\left[\left(\frac{z_\alpha}{\sqrt{T}}+\norm{\sigma_{22}^{-1}b_2}\right)^+\right]^2.\\
\lim_{\sigma_{11} \to \infty} \text{Var}(\log X^{\pi^*_c}(T)) &= T\left[\left(\frac{z_\alpha}{\sqrt{T}}+\sqrt{1-\delta^2}\norm{\sigma_{22}^{-1}b_2}\right)^+\right]^2.
\end{align}

%%%%%%%%%%%%%%%%%%%%%%%%%%%%%%%%%%%%%%%%%%%%%%%%%%
\subsection{Numerical Experiments}
In this section we consider some numerical examples to shed light in the portfolio diversification achieved by imposing a correlation constraint. In running our experiments we employ the market data of \cite{dmitravsinovic2011dynamic}.  Consequently, the market consists of three stocks and one of them is first type stock. For all proceeding examples, we consider the stock returns to have the constant standard deviations as below:
\begin{align}
\nonumber
\gamma_1 = 0.2, \quad \gamma_2 = 0.25, \quad \gamma_3 = 0.3
\end{align} 
%and the interest rate is $r=0.05$. 
\begin{example}
Here we focus on exploring the effect of market asset volatility on the variance of the log return. From the correlation matrix of stocks we find its corresponding volatility matrix \eqref{volmat} by Cholesky decomposition. Next by increasing $\sigma_{11}$ we track the behaviour of log return variances (these are seen as measures of diversification). Two sets of correlation matrices and excess mean returns of stocks as in \cite{dmitravsinovic2011dynamic} are investigated:
\begin{align}
\label{firstset}
\rho^{(1)} = \begin{bmatrix}
1.0 & -0.6 & -0.8 \\
-0.6 & 1.0 & 0.5 \\
-0.8 & 0.5 & 1.0
\end{bmatrix}, \quad
b^{(1)} = \begin{bmatrix} 0.07 & 0.05 & 0.03 \end{bmatrix} '
\end{align}

\begin{align}
\label{secondset}
\rho^{(2)} = \begin{bmatrix}
1.0 & -0.3 & 0.5\\
-0.3 & 1.0 & -0.9\\
0.5 & -0.9 & 1.0
\end{bmatrix}, \quad
b^{(2)} = \begin{bmatrix}0.03 & 0.05 & 0.07 \end{bmatrix} '
\end{align}

In Fig. \ref{fig:var}, the plots of log return variances for these two sets of data are presented, wherein both the associated graphs for the constrained problem are depicted for three different values of the correlation threshold $\delta.$ Notice that the highest curves in each plot corresponds to the unconstrained optimal portfolio. All the graphs are plotted for the fixed values of time horizon $T=5$ and confidence level $\alpha = 0.05.$
It's worth to look at the effect of $\delta$ on the variance; higher values of $\delta$ lead to more diversification (lower variance). Extreme situations may happen: notice that for $\delta=0.9$ in Fig. \ref{fig:Var1} there is zero variance, which means pure risk free investment. The diversification of the constrained versus unconstrained optimal portfolios is presented in both figures.
\begin{figure}
\centering
\begin{subfigure}{0.5\textwidth}
  \centering
  \includegraphics[width=\linewidth]{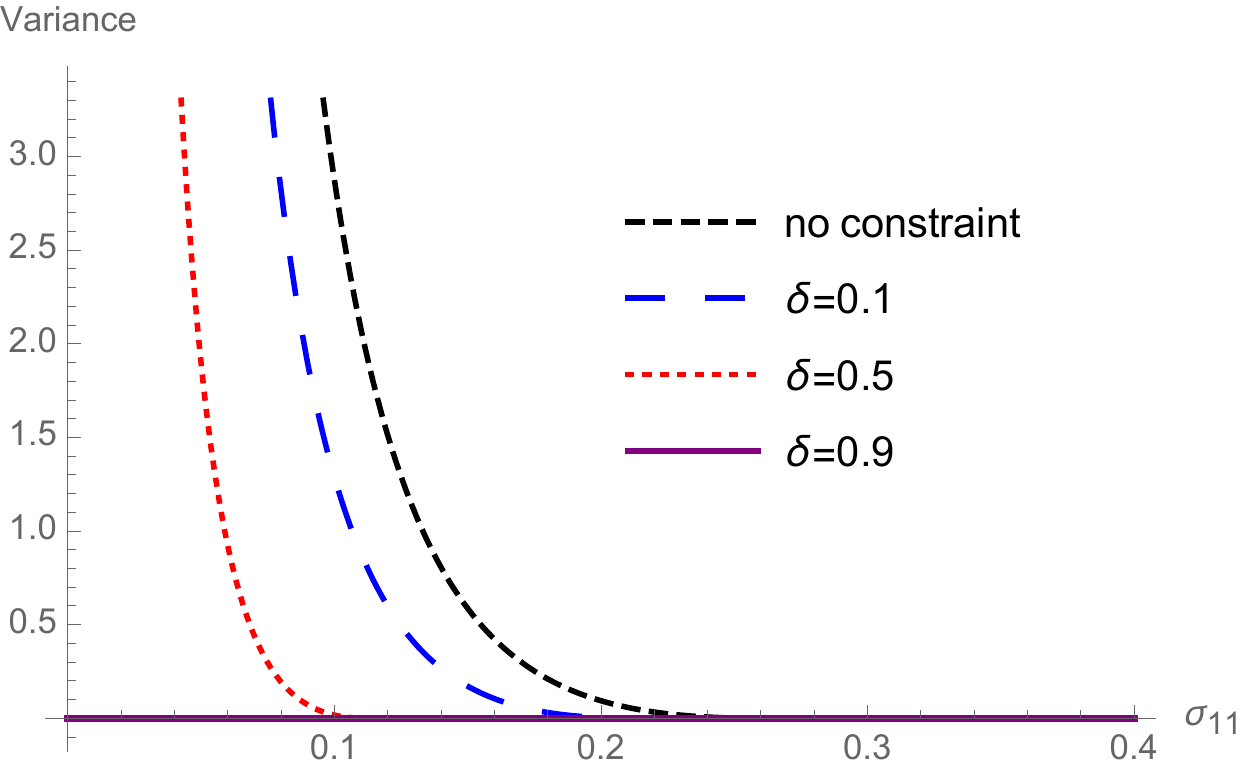}
  \caption{variance for 1st set of data}
  \label{fig:Var1}
\end{subfigure}%
\begin{subfigure}{.5\textwidth}
  \centering
  \includegraphics[width=\linewidth]{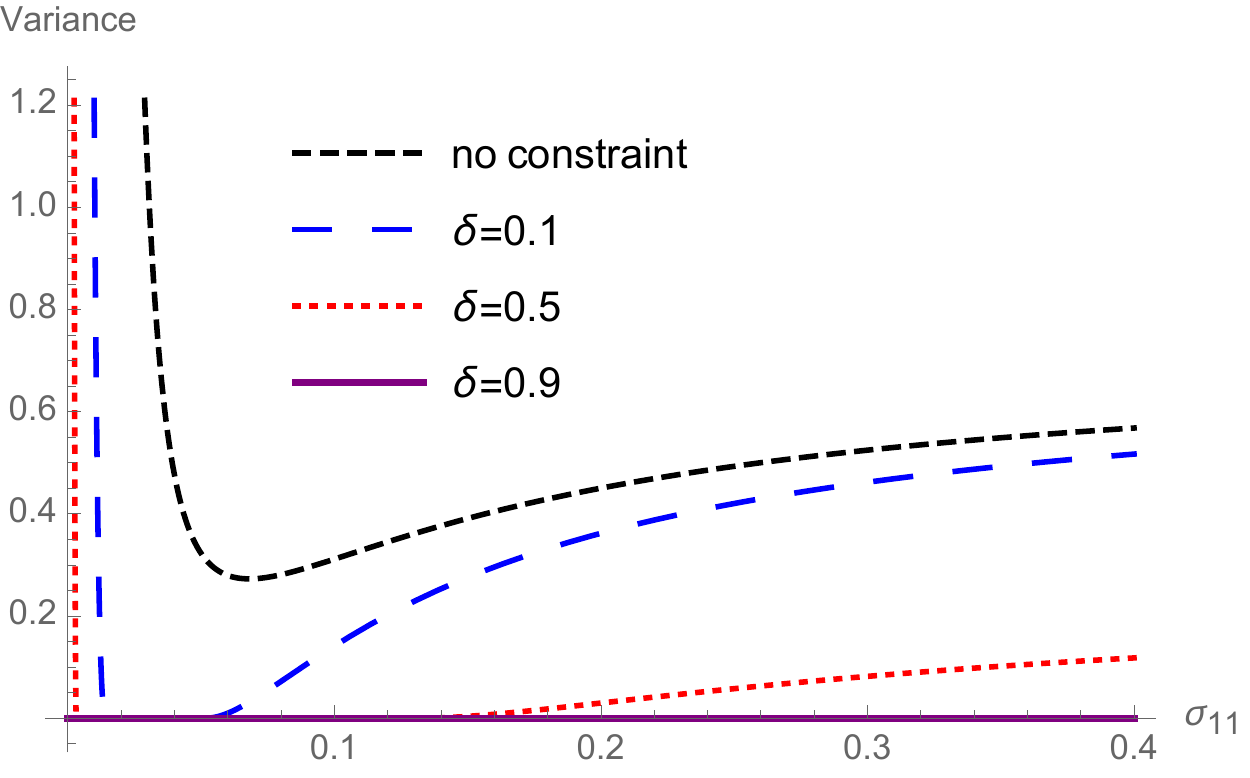}
  \caption{variance for 2nd set of data}
  \label{fig:Var2}
\end{subfigure}
\caption{Log return variance}
\label{fig:var}
\end{figure}
\end{example}
\begin{example}
In this example, we illustrate how the fraction of investment on risk less asset, i.e $\pi_0$, is changing as a response to an increase in market volatility $\sigma_{11}$. The findings are presented for both data sets of \eqref{firstset} and \eqref{secondset}. Since we didn't assume any restriction on borrowing/shortselling, negative values occur in some instances for optimal proportion of riskless investment in both plots in Fig. \ref{fig:pi0}. Let us notice from both graphs that the bigger $\delta$ the higher the riskless investment (for $\delta=0.9$ there is no investment on stocks and all portfolio is invested in the riskless asset, which is why the log return has zero variance in this case). One could also observe the pattern of investing more on riskless asset in Fig. \ref{fig:pi01} as a consequence of increase in market volatility, regardless of $\delta$. This observation does not occur in Fig. \ref{fig:pi02}, because of the structure of stock correlation matrix.
\begin{figure}
\centering
\begin{subfigure}{0.5\textwidth}
  \centering
  \includegraphics[width=\linewidth]{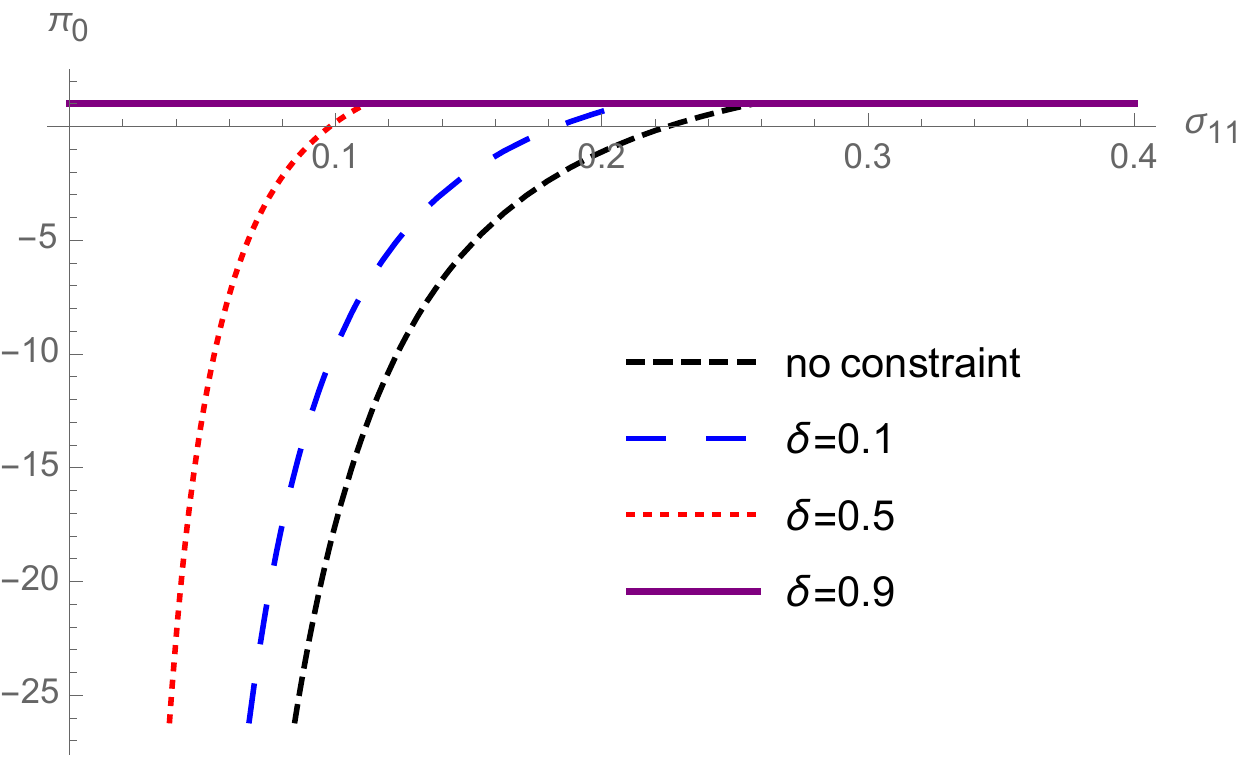}
  \caption{$\pi^*_{0,c}$ for 1st set of data}
  \label{fig:pi01}
\end{subfigure}%
\begin{subfigure}{.5\textwidth}
  \centering
  \includegraphics[width=\linewidth]{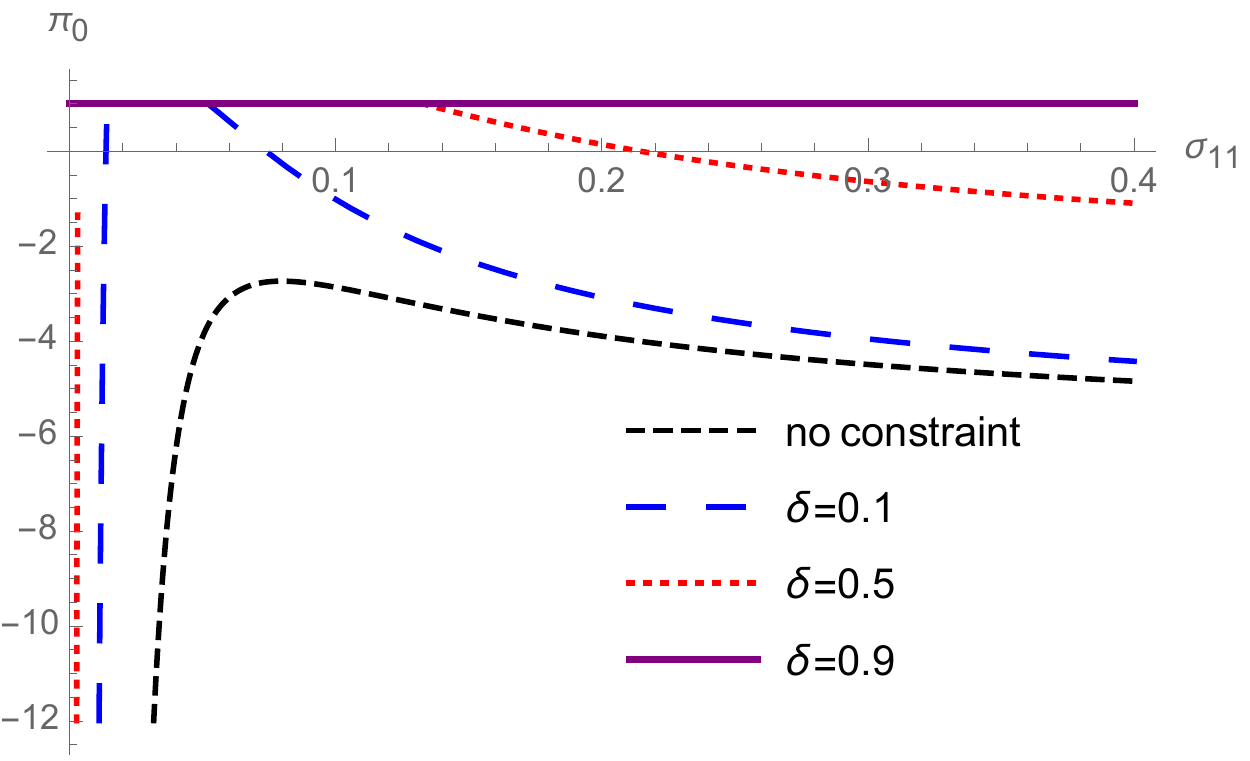}
  \caption{$\pi^*_{0,c}$ for 2nd set of data}
  \label{fig:pi02}
\end{subfigure}
\caption{Risk-less investment fraction}
\label{fig:pi0}
\end{figure}
\end{example}

\begin{example}
In the two previous examples, it is illustrated that an increase in $\delta$ leads to a better diversified portfolio. In this example we want to investigate this effect more precisely. Fig. \ref{fig:Varred} shows the percentage of variance reduction from unconstrained log return due to the correlation constraint. Both graphs show that by increasing $\delta$ the reduction in variance increases. The dotted line draws the $50\%$ variance reduction, which intercepts the curves at higher values of $\delta$, as we consider more volatile cases in the second sets of data. \begin{figure}
\centering
\begin{subfigure}{0.5\textwidth}
  \centering
  \includegraphics[width=\linewidth]{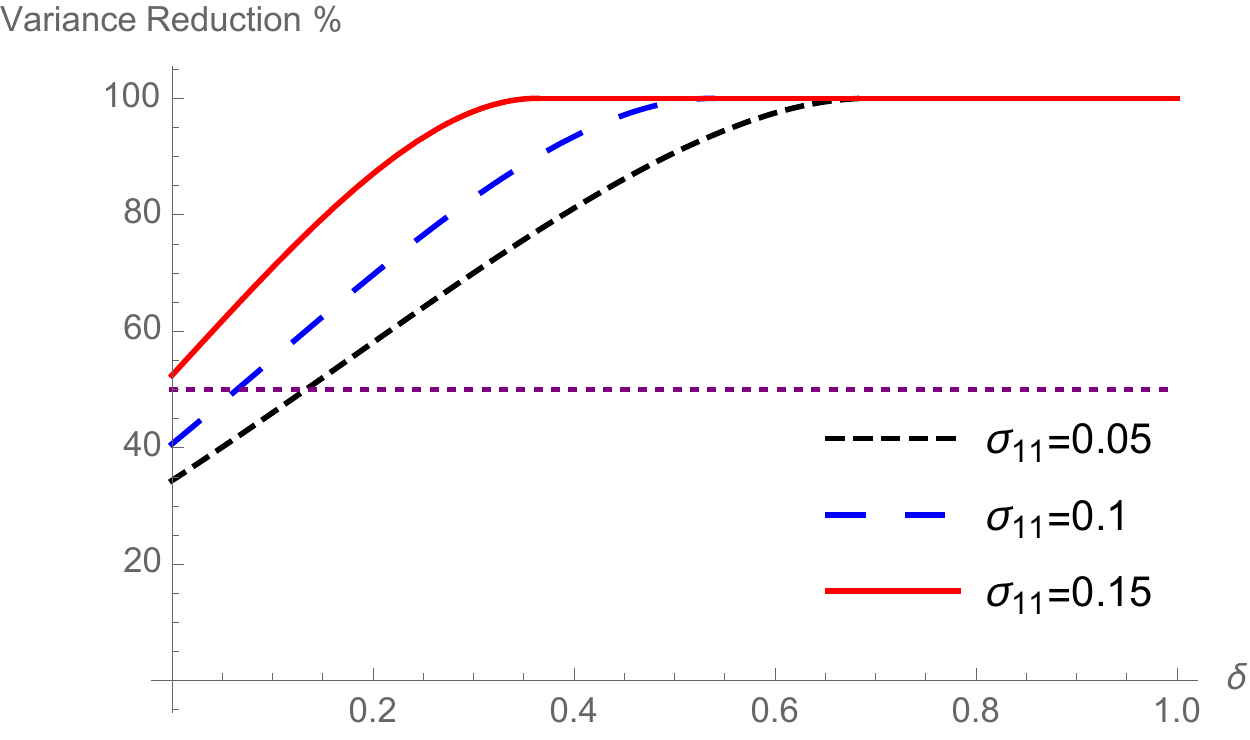}
  \caption{Variance reduction for 1st set of data}
  \label{fig:Varred1}
\end{subfigure}%
\begin{subfigure}{.5\textwidth}
  \centering
  \includegraphics[width=\linewidth]{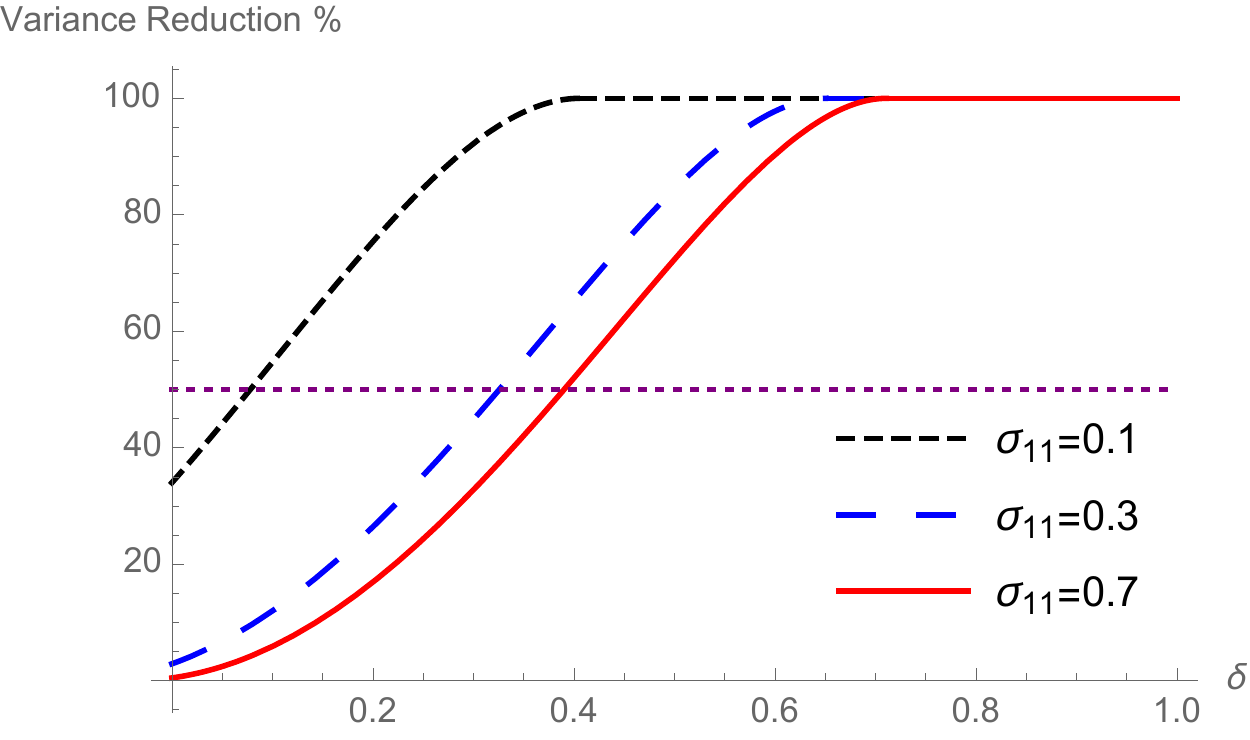}
  \caption{Variance reduction for 2nd set of data}
  \label{fig:Varred2}
\end{subfigure}
\caption{Variance reduction percentage}
\label{fig:Varred}
\end{figure}
\end{example}
%%%%%%%%%%%%%%%%%%%%%%%%%%%%%%%%%%%%%%%%%%%%%%%
\section{Conclusion}
In the Black-Scholes setting with constant parameters, the optimal portfolio which minimizes the capital at risk, and achieves a negative prescribed correlation with the given financial index is analytically derived. Moreover, it is shown that under a special choice of the index the correlation constraint
leads to a more diversified portfolio (the terminal variance of constrained optimal portfolio is lower compared to the optimal unconstrained portfolio). We also show how the correlation constraint reduces the variance and increases the risk-free investment during market downfalls. Numerical experiments explore the effect on the optimal portfolio composition induced by the correlation constraint.
%%%%%%%%%%%%%%%%%%%%%%%%%%%%%%%%%%%%%%%%%%%%%%%%%
\label{conclusion}

\section{Appendix}

\begin{proof}[Proof of Proposition \ref{CaR-min}] The analysis is a two phase procedure: first find the optimal portfolio vector on the boundary of ellipse $\norm{\sigma' \pi} = \varepsilon$, and then find the optimal ellipse parameter $\varepsilon$. Let's restrict the optimization domain of the minimization problem to an ellipse boundary, $\norm{\sigma' \pi}=\varepsilon$, and find the optimal portfolio vector on this set. Thus, in the first step the objective function is:
\begin{align}
\label{CaRe}
\text{CaR}(\pi,\alpha,T) = -b'\pi T+\frac{1}{2}\varepsilon^2T-z_\alpha \varepsilon \sqrt{T}.
\end{align}
We have to maximize the linear term $b'\pi$ over the boundary of the ellipsoid, in order to get the minimal CaR value on this set. By Cauchy-Schwarz inequality
\begin{align}
\nonumber
b'\pi = (\sigma^{-1}b)'(\sigma' \pi) &\leq \norm{\sigma^{-1}b} \norm{\sigma' \pi}\\
&= \norm{\sigma^{-1}b} \varepsilon.
\end{align}
The equality is attained for $\sigma' \pi = \frac{\varepsilon}{\norm{\sigma^{-1}b}} \sigma^{-1}b$ which leads to the optimal:
\begin{align}
\pi^*_{\varepsilon} = \frac{\varepsilon}{\norm{\sigma^{-1}b}} (\sigma \sigma')^{-1}b.
\end{align}
Substituting this choice of portfolio into \eqref{CaRe} leads to:
\begin{align}
\nonumber
\text{CaR}(\pi^*_{\varepsilon},\alpha,T) &= \frac{\varepsilon^2T}{2}-\varepsilon T\left(\frac{z_\alpha}{\sqrt{T}}+\norm{\sigma^{-1}b}\right)\\
&= \frac{\varepsilon T}{2} \left[ \varepsilon-2\left(\frac{z_\alpha}{\sqrt{T}}+\norm{\sigma^{-1}b}\right)\right].
\end{align}
which is minimized in $\varepsilon$ by
\begin{equation}
\varepsilon^*= \left(\frac{z_\alpha}{\sqrt{T}}+\norm{\sigma^{-1}b}\right)^+.
\end{equation}
Which concludes the equations \eqref{opt-pi} and \eqref{opt-CaR}.
\end{proof}

\begin{proof}[Proof of Theorem \ref{main}]
The objective function 
$$ \left(-b'\pi T+\frac{1}{2}\norm{\sigma' \pi}^2T-z_\alpha \norm{\sigma' \pi} \sqrt{T}\right)$$
is convex because of assumption $z_\alpha\leq 0.$
The constraint set
$$\delta \norm{\sigma' \eta}\norm{\sigma' \pi}+\pi'\sigma \sigma' \eta \leq 0$$
is a convex set in $\mathbb{R}^d$ being the level set of a convex function. In the light of the assumption $b'\eta > 0,$ the global minimum of this objective
function is not in the constraint set. Therefore the minimum is attained on the boundary, hence the optimization problem is equivalent to  
\begin{equation}
\label{optprob1}
\begin{aligned}
\min ~ &\left(-b'\pi T+\frac{1}{2}\norm{\sigma' \pi}^2T-z_\alpha \norm{\sigma' \pi} \sqrt{T}\right)\\
\text{subject to} ~~ & \delta \norm{\sigma' \eta}\norm{\sigma' \pi}+\pi'\sigma \sigma' \eta=0.
\end{aligned}
\end{equation}
Next, let us introduce the Lagrangian:
\begin{align}
L(\pi,\lambda) &= -b'\pi T+\frac{1}{2}\norm{\sigma' \pi}^2T-z_\alpha \norm{\sigma' \pi} \sqrt{T} + \lambda(\delta \norm{\sigma' \eta}\norm{\sigma' \pi}+\pi'\sigma \sigma' \eta).
\end{align}
The Lagrangian is minimized in a two phase procedure: first find the optimal $\pi$ on a boundary of an ellipse, and then find the optimal ellipse parameter.

\begin{align}
\label{2step}
\nonumber
\min_{ \pi} L(\pi,\lambda) &= \min_{\varepsilon \geq 0} ~ \min_{\norm{\sigma' \pi}=\ve} L(\pi,\lambda)\\
&= \min_{\varepsilon \geq 0} ~\min_{\norm{\sigma' \pi}=\ve} \left[-b'\pi T+\frac{1}{2}\ve^2T-z_\alpha \ve \sqrt{T} +\lambda(\delta \norm{\sigma' \eta}\ve+\pi'\sigma \sigma' \eta)\right].
\end{align} 
In the first step one has to solve:
\begin{align}
\text{maximize} \quad \left[ b'\pi T - \lambda \eta' \sigma \sigma' \pi\right] ~ \text{subject to} \quad \norm{\sigma' \pi} =\ve.
\end{align}
By Cauchy-Schwarz inequality, we have:
\begin{align}
\nonumber
 b'\pi T - \lambda \eta' \sigma \sigma' \pi  &= (\sigma^{-1}bT - \lambda \sigma' \eta)' (\sigma' \pi)\\
 & \leq \norm{\sigma^{-1}bT - \lambda \sigma' \eta} \ve.
\end{align}
The equality occurs at $\sigma ' \pi_\ve = \frac{\ve}{\norm{\sigma^{-1}bT - \lambda \sigma' \eta}} (\sigma^{-1}bT - \lambda \sigma' \eta)$, which leads to:
\begin{align}
\label{pive}
\pi_\ve = \frac{\ve}{\norm{\sigma^{-1}bT - \lambda \sigma' \eta}} \left((\sigma \sigma')^{-1}bT-\lambda\eta\right).
\end{align}
By substituting \eqref{pive} back into the \eqref{2step}, the second minimization problem will reduce to:
\begin{align}
\label{quadinf}
\nonumber
\min_{\ve \geq 0} \quad & \left[\frac{1}{2}\ve^2T-z_\alpha \ve \sqrt{T}+ \lambda \delta \ve \norm{\sigma' \eta}-\ve \norm{\sigma^{-1}bT - \lambda \sigma' \eta} \right]
\end{align}
which is solved by
\begin{equation}
\label{optve}
\ve^*(\lambda) = \frac{1}{T} \left(z_\alpha \sqrt{T} -\lambda \delta \norm{\sigma'\eta}+\norm{\sigma^{-1}bT - \lambda \sigma' \eta}\right)^+.
\end{equation}
The optimal portfolio is
\begin{align}
\label{optpilam}
\pi^* = \frac{\ve^*(\lambda^*)}{\norm{\sigma^{-1}bT - \lambda^* \sigma' \eta}} \left((\sigma \sigma')^{-1}bT-\lambda^*\eta\right),
\end{align}
where $\lambda^*$ is derived from

\begin{align}
 \label{cs}
 \lambda^* (\delta \norm{\sigma' \eta}\norm{\sigma'\pi^*}+\pi^{*'}\sigma\sigma'\eta) = 0.
 \end{align}
By direct computations this becomes
\begin{align}
\lambda^{*^2} \norm{\sigma' \eta}^4-2\lambda^*\norm{\sigma' \eta}^2 b' \eta T+\frac{T^2}{1-\delta^2}\left((b' \eta)^2-\delta^2\norm{\sigma' \eta}^2\norm{\sigma^{-1}b}^2\right)=0,
\end{align}
which has the positive solution
\begin{align}
\lambda^* = \frac{1}{\norm{\sigma' \eta}^2} \left(b' \eta T+\frac{T\delta}{\sqrt{1-\delta^2}}\sqrt{\norm{\sigma^{-1}b}^2\norm{\sigma' \eta}^2 -(b' \eta)^2}\right).
\end{align}
 \end{proof}

\begin{proof}[Proof of Proposition \ref{pk}] 
Proof follows from Theorem \ref{main}. Given the volatility matrix and its inverse of \eqref{volmat}, and the prescribed portfolio vector for the benchmark process, one can readily find:
\begin{align}
\norm{\sigma' \eta} &= \norm{\sigma_{11}^{-1}b_1}.\\
b'\eta &= \norm{\sigma_{11}^{-1}b_1}^2.\\
\sqrt{\norm{\sigma^{-1}b}^2\norm{\sigma' \eta}^2 -(b' \eta)^2} &= \norm{\sigma_{11}^{-1}b_1}\norm{\sigma_{22}^{-1}b_2-\sigma_{22}^{-1}\sigma_{21}\sigma_{11}^{-1}b_1}.
\end{align}
By plugging these equalities into $\pi^*$ and $\lambda^*$ of Theorem \ref{main}, the equations \eqref{optpipk} and \eqref{optlambdapk} yield.
\end{proof}

\begin{proof}[Proof of Proposition \ref{var-reduction}]
Let us denote by 
$$\theta_1 = \norm{\sigma_{11}^{-1}b_1},$$
and
$$\theta_2= \norm{\sigma_{22}^{-1}b_2-\sigma_{22}^{-1}\sigma_{21}\sigma_{11}^{-1}b_1}.$$
Direct computations lead to:
\begin{align}
\text{Var}(\log X^{\pi^*}(T)) &= T\left[\left(\frac{z_\alpha}{\sqrt{T}}+\sqrt{\theta_2^2+\theta_1^2}\right)^+\right]^2.\\
\text{Var}(\log X^{{\pi_{c}}^*}(T)) &= T\left[\left(\frac{z_\alpha}{\sqrt{T}} +\sqrt{1-\delta^2}\theta_2 -\delta \theta_1\right)^+\right]^2.
\end{align}
In the light of the inequality:
\begin{align}
\nonumber
\sqrt{\theta_2^2+\theta_1^2} \geq \sqrt{1-\delta^2}\theta_2 -\delta \theta_1,
\end{align}
the claim yields.
\end{proof}

\bibliographystyle{unsrt}
\bibliography{ref}

\begin{thebibliography}{1}

\bibitem{markowitz1991foundations}
Harry~M Markowitz.
\newblock Foundations of portfolio theory.
\newblock {\em The journal of finance}, 46(2):469--477, 1991.

\bibitem{duffie1997overview}
Darrell Duffie and Jun Pan.
\newblock An overview of value at risk.
\newblock {\em The Journal of derivatives}, 4(3):7--49, 1997.

\bibitem{jorion2007value}
Philippe Jorion.
\newblock {\em Value at risk: the new benchmark for managing financial risk},
  volume~2.
\newblock McGraw-Hill New York, 2007.

\bibitem{gambrah2014risk}
Priscilla Serwaa~Nkyira Gambrah and Traian~Adrian Pirvu.
\newblock Risk measures and portfolio optimization.
\newblock {\em Journal of Risk and Financial Management}, 7(3):113--129, 2014.

\bibitem{emmer2001optimal}
Susanne Emmer, Claudia Kl{\"u}ppelberg, and Ralf Korn.
\newblock Optimal portfolios with bounded capital at risk.
\newblock {\em Mathematical Finance}, 11(4):365--384, 2001.

\bibitem{dmitravsinovic2011dynamic}
Gordana Dmitra{\v{s}}inovi{\'c}-Vidovi{\'c}, Ali Lari-Lavassani, Xun Li, and
  Antony Ware.
\newblock Dynamic portfolio selection under capital-at-risk with no
  short-selling constraints.
\newblock {\em International Journal of Theoretical and Applied Finance},
  14(06):957--977, 2011.

\bibitem{dmitravsinovic2006asymptotic}
Gordana Dmitra{\v{s}}inovi{\'c}-Vidovi{\'c} and Antony Ware.
\newblock Asymptotic behaviour of mean-quantile efficient portfolios.
\newblock {\em Finance and Stochastics}, 10(4):529--551, 2006.

\bibitem{bernard2014mean}
Carole Bernard and Steven Vanduffel.
\newblock Mean--variance optimal portfolios in the presence of a benchmark with
  applications to fraud detection.
\newblock {\em European Journal of Operational Research}, 234(2):469--480,
  2014.

\end{thebibliography}

\end{document}